\documentclass[11pt,a4paper,onecolumn]{IEEEtran}

\usepackage{graphicx}
\usepackage{amsmath, amsthm, amssymb, amsfonts, cancel}
\usepackage{algorithm, algorithmic, ifsym, subfigure}

\newcommand\blfootnote[1]{%
  \begingroup
  \renewcommand\thefootnote{}\footnote{#1}%
  \addtocounter{footnote}{-1}%
  \endgroup
}

\def\bd{ \textbf{d} } 
 
\def\bp{ \textbf{p} } 
\def\ba{ \textbf{a} } 
\def\bA{ {\mathbf{A}} }

\def\bD{ {\mathbf{D}} }

\def\bI{ {\mathbf{I}} }

\def\bP{ {\mathbf{P}} }

\def\bepsilon{ {\mathbf{\epsilon}} }

\def\bPhi{ {\mathbf{\Phi}} }

\def\b1{ {\mathbf{1}} }

\def\nn{{ \parallel   }}
\def\RR{{ \mathbb{R}  }}
\def\PP{{ \mathbb{P}  }}
\def\EE{{ \mathbb{E}  }}
\def\NN{{ \mathbb{N}  }}

\def\bx{{ \mathbf{x}  }}
\def\by{{ \mathbf{y}  }}

\def\be{{ \mathbf{e}  }}
\def\bv{{ \mathbf{v}  }}
\def\bc{{ \mathbf{c}  }}

\newtheorem{theorem}{Theorem}
\newtheorem{lemma}{Lemma}

\newtheorem{proposition}{Proposition}
\newtheorem{assumption}{Assumption}

\begin{document}
\title{Decentralized Adaptive Search using the Noisy 20 Questions Framework in Time-Varying Networks}
\author{Theodoros Tsiligkaridis, \textit{Member, IEEE}}

\maketitle

\begin{abstract}
This paper considers the problem of adaptively searching for an unknown target using multiple agents connected through a time-varying network topology. Agents are equipped with sensors capable of fast information processing, and we propose a decentralized collaborative algorithm for controlling their search given noisy observations. Specifically, we propose decentralized extensions of the adaptive query-based search strategy that combines elements from the 20 questions approach and social learning. Under standard assumptions on the time-varying network dynamics, we prove convergence to correct consensus on the value of the parameter as the number of iterations go to infinity. The convergence analysis takes a novel approach using martingale-based techniques combined with spectral graph theory. Our results establish that stability and consistency can be maintained even with one-way updating and randomized pairwise averaging, thus providing a scalable low complexity method with performance guarantees. We illustrate the effectiveness of our algorithm for random network topologies.
\end{abstract}


\blfootnote{The material in this paper was presented in part at the 2016 IEEE International Symposium on Information Theory (ISIT), Barcelona, Spain.

T. Tsiligkaridis is with MIT Lincoln Laboratory, Lexington, MA 02421 USA (email: ttsili@ll.mit.edu).
}

\section{Introduction} \label{sec:intro}
Consider a set of agents that try to estimate a parameter, e.g., estimate a target state or location, collectively. The agents are connected by a time-varying information sharing network and can periodically query one of their local neighbors about the target location. In this paper we adopt a generic observation model based on query-response models where the queries are functions of agents' local information and successive queries are determined by a feedback control policy. Specifically, in the 20 questions-type model considered in this paper, the observation of each agent is coupled with the query region chosen by that agent, which is a function of its current local belief.

A centralized collaborative 20 questions framework was proposed and studied in \cite{Tsiligkaridis:TIT:2014}, where a global centralized controller jointly or sequentially formulates optimal queries about target location for all agents. This work was later extended in \cite{OnDecEstActQ:2015} to the decentralized setting, in which each agent formulates his own query based on his local information and exchanges beliefs with its neighbors in a synchronous fashion (i.e, updating the beliefs of all agents simultaneously at each update step). The proposed decentralized algorithm therein consisted of two stages: 1) local belief update; and 2) local information sharing. In stage 1 each agent implements the bisection query policy of \cite{Tsiligkaridis:TIT:2014} to update their local belief function. In stage 2 the local belief functions are averaged over nearest neighborhoods in the information sharing network. This two-stage algorithm was proven to converge to a consensus estimate of the true state, assuming synchronous updating of all agents' beliefs in the network and irreducibility of the social interaction graph.

The decentralized collaborative 20 questions problem is applicable to large scale collaborative stochastic search applications where there is no centralized authority. Examples include: object tracking in camera networks \cite{Sznitman:2010}; road tracking from satellite remote sensing networks \cite{Geman:1996}; and wide area surveillance networks \cite{CastroNowak07}. Other applications may include extending active testing approaches in the decentralized setting for classification problems, for instance in vision, recommendation systems, and epidemic networks. The 20 questions paradigm is motivated by asking the correct type of questions in the correct order and is applicable to various other domains where computational effort and time are critical resources to manage.

In this paper, we consider a variant of the two-stage decentralized collaborative algorithm of \cite{OnDecEstActQ:2015} by relaxing the assumption of a fixed network topology. We consider time-varying network topologies in which two randomly chosen agents interact at each update step, giving rise to asynchronous updates (i.e., beliefs of all agents are not updated at each update step). We analyze the convergence properties of the two-stage asynchronous decentralized collaborative 20 questions algorithm under appropriate conditions. This asynchronous model is applicable to practical sensor networks where agents may be located in large geometric distances, and as a result their wireless communications are unreliable or intermittent due to path occlusions or other environmental effects, and lead to time-varying network topologies. Our analysis is based on smoothing techniques and martingale convergence theory in similar spirit to \cite{OnDecEstActQ:2015}. However, the randomness due to the time-varying network topology and lack of strong connectivity at each time instant introduce additional complications in the analysis that were absent when analyzing the static network case of adaptive agents in \cite{OnDecEstActQ:2015}.

In addition to theoretical analysis of the convergence of the proposed algorithm, numerical studies of performance are provided showing interesting information behavior that information sharing yields. The benefit of our asynchronous approach is that the asynchronous decentralized algorithm attains similar performance as its synchronous counterpart introduced and analyzed in \cite{OnDecEstActQ:2015}. The lack of synchronization our approach offers is of great practical interest because synchronization of a large number of agents can be difficult \cite{Wu:2011}. Furthermore, the popular time-division-multiple-access (TDMA) communication protocol for distributed networks is only applicable to synchronized networks. Even though the synchronous update scheme in \cite{OnDecEstActQ:2015} is proven to be convergent to the correct limit, it is still an open question whether the same algorithm converges with an asynchronous implementation. Counterexamples showing that asynchronous updates do not converge although synchronous updates converge are presented in \cite{Xia:2014} for the standard consensus problem.

\subsection{Prior Work}
The noisy 20 questions problem, also known as Ulam's game, was introduced by Renyi \cite{Renyi:1961} and was later rediscovered by Ulam \cite{Ulam:1976}. The first work making the connection between communication with feedback and noisy search appeared in \cite{BZ}. The probabilistic bisection algorithm dates back to the work of Horstein \cite{Horstein:1963}, where it was originally proposed and analyzed heuristically in the contest of communication with noiseless feedback over the binary symmetric channel. This algorithm was shown to achieve capacity for arbitrary memoryless channels in \cite{Shayevitz:2007, Shayevitz:2011}. The probabilistic bisection algorithm was generalized to multiple players in \cite{Tsiligkaridis:TIT:2014} in the centralized setting, and decentralized algorithms for probabilistic bisection search were proposed in \cite{OnDecEstActQ:2015}.

Our work also differs from the works on 20 questions/active stochastic search of Jedynak, et al., \cite{Jedynak12}, Castro \& Nowak \cite{CastroNowak07}, Waeber, et al., \cite{Waeber:2013}, and Tsiligkaridis, et al., \cite{Tsiligkaridis:TIT:2014} because we consider intermediate local belief sharing between agents after each local bisection and update. In addition, in contrast to previous work, in the proposed framework each agent incorporates the beliefs of its neighbors in a way that is agnostic of its neighbors' error probabilities. The analysis of \cite{Jadbabaie:2012, Molavi:2013} does not apply to our model since we consider controlled observations, although we use a form of the social learning model of \cite{Jadbabaie:2012, Molavi:2013}. While a randomized distributed averaging/consensus problem was analyzed in \cite{Boyd:2006}, the convergence analysis is not applicable because we consider new information injected in the dynamical system at each iteration (controlled information gathering) in addition to randomized information sharing.

Although consensus to the true target location holds for the degenerate case of no agent collaboration by using results in the existing literature, collaboration improves the rate of convergence of the estimation error. As shown in the numerical results in this paper, the error decays faster as a function of iterations. This is the primary motivation for studying such collaborative signal processing algorithms. However, proving convergence to the correct consensus is the first step in analyzing such algorithms, and it is by no means a trivial one. In fact, even for the simple single-agent case, only very recently, Waeber, et al \cite{Waeber:2013} were able to prove the first rigorous convergence rate result for the continuous probabilistic bisection algorithm. The focus of this paper is to establish convergence of decentralized algorithms for probabilistic bisection in time-varying networks.


\section{Notation} \label{sec:notation}
We define $X^*$ the true parameter, the target state in the sequel, and its domain as the unit interval $\mathcal{X}=[0,1]$. Let $\mathcal{B}(\mathcal{X})$ be the set of all Borel-measurable subsets $B \subseteq \mathcal{X}$. Let $\mathcal{N}=\{1,\dots,M\}$ index the $M$ agents in an interaction network, denoted by the vertex set $\mathcal{N}$ and the directed edges joining agents at time $t \in \NN$ are captured by $E(t)$. Let $\bA_t=\{a_{i,j}(t)\}$ denote the interaction matrix at time $t$, which is a stochastic matrix (i.e., nonnegative entries with rows summing to unity). At each time $t$, the time-varying network structure is modeled by the directed graph $(\mathcal{N},E(t))$, where
\begin{equation*}
	E(t) = \{(j,i): [\bA_t]_{i,j} > 0 \}
\end{equation*}
Let $\bA_{i\to j}$ denote the interaction matrix when agent $i$ performs a Bayesian update based on its query and averages beliefs with agent $j$. In our model, a random agent $i$ is chosen with probability $q_i$ and a collaborating agent $j$ is chosen with probability $P_{i,j}$ at each update step. Thus, the interaction matrix $\bA_t = \bA_{i\to j}$ is chosen with probability $q_i P_{i,j}$, and nodes $i$ and $j$ collaborate. The matrix $\bP=\{P_{i,j}\}$ contains the probabilistic weights for collaboration between agents and are zero when there is no edge in the information sharing network at any time; if $P_{i,j}=0$, then agent $i$ cannot collaborate with agent $j$ at any time.

Define the probability space $(\Omega, \mathcal F, \mathbb P)$ consisting of the sample space $\Omega$ generating the unknown state $X^*$ and the observations $\{Y_{i,t+1}\}$ at times $t=0, 1, \ldots$, an event space $\mathcal F$ and a probability measure $\mathbb P$. The expectation operator $\mathbb E$ is defined with respect to $\mathbb P$.

Define the belief of the $i$-th agent at time $t$ on $\mathcal{X}$ as the posterior density $p_{i,t}(x)$  of target state $x\in \mathcal X$ based on all of the information available to this agent at this time. Define the $M\times 1$ vector $\bp_t(x)=[p_{1,t}(x),\dots,p_{M,t}(x)]^T$ for each $x \in \mathcal{X}$. For any $B\in \mathcal{B}(\mathcal{X})$, define $\bP_t(B)$ as the vector of probabilities with $i$-th element equal to $\int_B p_{i,t}(x) dx$. We define the query point/target estimate of the $i$-th agent as $\hat{X}_{i,t}$. The query point is the right boundary of the region $A_{i,t}=[0,\hat{X}_{i,t}]$. We let $F_{i,t}(a) = P_{i,t}([0,a])=\int_{0}^a p_{i,t}(x) dx$ denote the cumulative distribution function associated with the density $p_{i,t}(\cdot)$.

We assume that a randomly chosen agent $i$ constructs a query at time $t$ of the form ``does $X^*$ lie in the region $A_{i,t} \subset \mathcal{X}$?''. We indicate this query with the binary variable $Z_{i,t}=I(X^* \in A_{i,t})$ to which agent $i$ responds with a binary response $Y_{i,t+1}$, which is correct with probability $1-\epsilon_i$, and without loss of generality $\epsilon_i \leq 1/2$. This error model is equivalent to a binary symmetric channel (BSC) with crossover probability $\epsilon_i$. The query region $A_{i,t}$ depends on the accumulated information up to time $t$ at agent $i$. Define the nested sequence of event spaces $\mathcal{F}_t$, $\mathcal{F}_{t-1} \subset \mathcal{F}_{t}$, for all $t\geq 0$, generated by the sequence of queries and responses. The queries $\{A_{i,t}\}_{t\geq 0}$ are measurable with respect to this filtration. Define the canonical basis vectors $\be_i\in \RR^M$ as $[\be_i]_j=I(j=i)$. The notation \textit{i.p.} denotes convergence in probability and \textit{a.s.} denotes almost-sure convergence.

\section{Asynchronous Decentralized 20 Questions}

Motivated by the work of \cite{Tsiligkaridis:TIT:2014, OnDecEstActQ:2015} and \cite{Jadbabaie:2012}, we proceed as follows. As in the fixed-topology decentralized algorithm in \cite{OnDecEstActQ:2015}, starting with a collection of prior distributions $\{p_{i,0}(x)\}_{i \in \mathcal{N}}$ on $X^*$, the goal is to reach consensus across the network through repeated querying and information sharing. Our proposed asynchronous decentralized collaborative 20 questions algorithm consists of two stages. Motivated by the optimality of the bisection rule for symmetric channels proved by Jedynak, et al., \cite{Jedynak12}, the first stage bisects the posterior of a randomly chosen agent $i\in \mathcal{N}$ (say with probability $q_i$) at $\hat{X}_{i,t}$ and refines its own belief through Bayes' rule \footnote{In the asynchronous time model of Boyd, et al., \cite{Boyd:2006}, each agent has a clock that ticks according to a rate-one Poisson process. As a result, the inner ticks at each agent $i$ are distributed according to a rate-one exponential distribution, independently across agents and over time. This corresponds to a single clock ticking according to a rate-$M$ Poisson process at times $\{t_k: k \geq 1\}$, where $\{\tilde{t}_k=t_{k+1}-t_k\}$ are i.i.d. exponential random variables of rate $M$. Let $i_k\in\mathcal{N}$ denote the agent whose clock ticked at time $t_k$. It follows that $i_k$ are i.i.d. $\text{Unif}(\mathcal{N})$, i.e., $q_i=1/M$.}. In the second stage, agent $i$ collaborates with an agent $j$ with probability $P_{i,j}$ by averaging their beliefs (see Algorithm 1).


Some simplifications occur in Algorithm 1. The normalizing factor $\mathcal{Z}_{i,t}(y)$ is given by $\int_{\mathcal{X}} p_{i,t}(x) l_i(y|x,\hat{X}_{i,t}) dx$ and can be shown to be equal to $1/2$ (see proof of Lemma \ref{lemma:lemmaA} in Appendix A). The bisection query points are medians $\hat{X}_{i,t}=F_{i,t}^{-1}(1/2)$ and the observation distribution is:
\begin{equation*}
	l_i(y|x,\hat{X}_{i,t}) = f_1^{(i)}(y) I(x \leq \hat{X}_{i,t}) + f_0^{(i)}(y) I(x > \hat{X}_{i,t}).
\end{equation*}
where the distributions $f_z^{(i)}(\cdot)$ are defined in (\ref{eq:BSC}). We note that the conditioning on the query region $A_{i,t}$ (or query point $\hat{X}_{i,t}$ in one-dimension) is necessary as the binary observation $\by$ is linked to the query in the 20 questions model, in which the correct answer is obtained with probability $1-\epsilon_i$ and the wrong answer is obtained with probability $\epsilon_i$.
We remark that the density $l_i(y|x,\hat{X}_{i,t})$ depends on the query point $\hat{X}_{i,t}$, which is time-varying and as a result, the density $l_i(y|x,\hat{X}_{i,t})$ is time-varying. 


\begin{algorithm}
\caption{ Asynchronous Decentralized Bisection Search Algorithm for Time-Varying Networks}
\label{alg:alg1}
\begin{algorithmic}[1]
\STATE \textbf{Input:}  {$\mathcal{N}, \bP=\{P_{i,j}: (i,j)\in \mathcal{N}\times \mathcal{N}\}, \{\epsilon_i: i\in \mathcal{N}\}$}
\STATE \textbf{Output:} {$\{\hat{X}_{i,t}:i\in \mathcal{N}\}$}

	\STATE Initialize $p_{i,0}(\cdot)$ to be positive everywhere.
	
	\REPEAT
		\STATE Choose an agent $i\in \mathcal{N}$ randomly (e.g., with probability $q_i$). \\
		\STATE	\quad Bisect posterior density at median: $\hat{X}_{i,t}=F_{i,t}^{-1}(1/2)$. \\
		\STATE	\quad Obtain (noisy) binary response $y_{i,t+1} \in \{0,1\}$. \\
		\STATE Choose a collaborating agent $j\in \mathcal{N}$ with probability $P_{i,j}$. \\
	  \STATE	\quad Belief update: \\
		\begin{align}
			\label{eq:randomized_density_update}
			\begin{split}
			p_{i,t+1}(x) &= \alpha_i p_{i,t}(x) \frac{l_i(y_{i,t+1}|x,\hat{X}_{i,t})}{\mathcal{Z}_{i,t}(y_{i,t+1})} + (1-\alpha_i) p_{j,t}(x), \\
			p_{j,t+1}(x) &= p_{i,t+1}(x), \\
			p_{l,t+1}(x) &= p_{l,t}(x), \forall l\neq i, l\neq j, \qquad x\in \mathcal{X}. 
			\end{split}
		\end{align}
		where the observation p.m.f. is:
		\begin{align*}
			l_i(y|x,\hat{X}_{i,t}) &= f_1^{(i)}(y) I(x \leq \hat{X}_{i,t}) + f_0^{(i)}(y) I(x > \hat{X}_{i,t}),\\
			&\qquad \qquad y \in \mathcal{Y}
		\end{align*}
		and $f_1^{(i)}(y)=(1-\epsilon_i)^{I(y=1)} \epsilon_i^{I(y=0)}, f_0^{(i)}(y)=1-f_1^{(i)}(y)$.
	\UNTIL {convergence}	
\end{algorithmic}
\end{algorithm}
We remark that Algorithm 1 is fully decentralized, i.e., only local information processing and local information sharing is needed. Furthermore, it operates in an asynchronous fashion as the belief updates are not occurring simultaneously for all agents in the network. At each step, two agents collaborate with each other and update their beliefs. This belief averaging leads to a non-Bayesian social learning scheme in similar spirit to \cite{Jadbabaie:2012, Molavi:2013}.

Collaboration becomes essential in certain scenarios when some agents in the network are completely unreliable $\epsilon_i = 1/2$, and helps tremendously when agents are unreliable, i.e., $\epsilon_i \approx 1/2$. When an agent is completely unreliable, its learning capacity is zero and cannot localize the target on its own. With collaboration, neighboring agents may steer the belief of unreliable agents closer to the true location (also see Section \ref{sec:simulations} for several experiments).

\section{Convergence of Asynchronous Decentralized Estimation Algorithm} \label{sec:convergence}

In this section convergence properties of Algorithm 1 are established under the assumptions below. The two main theoretical results,  Thm. 1 and Thm. 2, establish that the proposed algorithm attains asymptotic agreement (consensus) and asymptotic consistency, respectively. Several technical lemmas are necessary and are proven in the appendices. A block diagram showing the interdependencies between the lemmas and theorems in this section is shown in Fig. \ref{fig:guide}.
\begin{figure}[ht]
	\centering
		\includegraphics[width=0.40\textwidth]{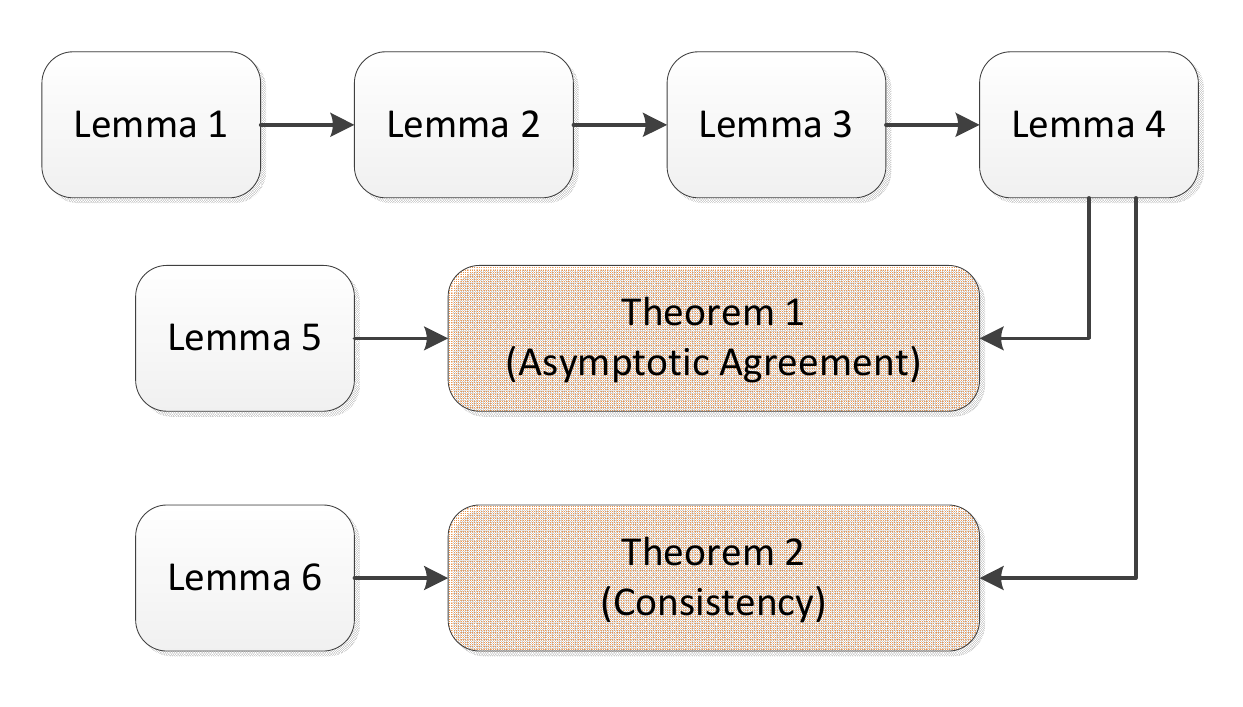}
	\caption{ The flow of the analysis for establishing convergence of the asynchronous decentralized 20-questions algorithm to the correct consensus limit. }
	\label{fig:guide}
\end{figure}

According to Algorithm \ref{alg:alg1}, the density evolution described by (\ref{eq:randomized_density_update}) can be written in matrix form as:
\begin{equation} \label{eq:rand_density_evol}
	\bp_{t+1}(x) = (\bA_t + \bD_t(x)) \bp_t(x)
\end{equation}
where $\bA_t$ accounts for the averaging action of the iteration and $\bD_t(x)$ is the innovation. According to Algorithm \ref{alg:alg1}, it follows that with probability $q_i P_{i,j}$, agents $i$ and $j$ collaborate, resulting in:
\begin{align}
	\bA_t    	&= \bA_{i \to j} = \bI_M + (\be_i+\be_j)(\alpha_{i}\be_i+(1-\alpha_{i})\be_j)^T \nonumber \\
						&\qquad \qquad \qquad - (\be_i\be_i^T+\be_j\be_j^T) \label{eq:rand_A} \\
	\bD_t(x) 	&= \bD_{i \to j}(x) = \alpha_i \left(\frac{l_i(Y_{i,t+1}|x,\hat{X}_{i,t})}{\mathcal{Z}_{i,t}(Y_{i,t+1})} - 1 \right) (\be_i+\be_j) \be_i^T \label{eq:rand_D}
\end{align}

\subsection{Assumptions}

To simplify the analysis of Algorithm \ref{alg:alg1}, we make the following assumptions, which are comparable to those made in \cite{Tsiligkaridis:TIT:2014, OnDecEstActQ:2015} and \cite{Jadbabaie:2012}.

Each agent's response is governed by the conditional distribution:
\begin{align}
	l_i(y_i|x,A_{i,t}) &\stackrel{def}{=} P(Y_{i,t+1}=y_i|A_{i,t},X^*=x) \nonumber \\
		&= \left\{ \begin{array}{ll} f_1^{(i)}(y_i), & x\in A_{i,t} \\ f_0^{(i)}(y_i), & x\notin A_{i,t} \end{array} \right.  \label{eq:exp}
\end{align}

\begin{assumption} (Memoryless Binary Symmetric Channels) \label{assump:BSC}
	We model the agents' responses as independent (memoryless) binary symmetric channels (BSC) \cite{CoverThomas} with crossover probabilities $\epsilon_i\in (0,1/2]$. The probability mass function $f_z^{(i)}(Y_{i,t+1}) = P(Y_{i,t+1}|Z_{i,t}=z)$ is:
\begin{equation} \label{eq:BSC}
	f_z^{(i)}(y_i) = \Bigg\{ \begin{array}{ll} 1-\epsilon_i, & y_i=z \\ \epsilon_i, & y_i\neq z \end{array}
\end{equation}
for $i=1,\dots,M, z\in \{0,1\}$. Define the set of agents $\mathcal{I}_1 = \{i\in \mathcal{N}: \epsilon_i < 1/2\}$ and assume the set $\mathcal{I}_1$ is nonempty.
\end{assumption}
The condition $\epsilon_i<1/2$ implies that the response of an agent $i$ is `almost correct'. Agents $i\notin \mathcal{I}_1$ do not have the ability to localize the target on their own.

\begin{assumption} (Average Strong Connectivity) \label{assump:strong_connectivity}
	As in \cite{Jadbabaie:2012}, we assume that the network is strongly connected on average, i.e. $\bA=\EE[\bA_t]=\sum_{i,j} q_i P_{i,j} \bA_{i\to j}$ is irreducible. Furthermore, we assume $\underline{\alpha} = \min_i \alpha_i >0$ to ensure collaboration.
\end{assumption}

The following standard assumption about the network connectivity over time \cite{Tsitsiklis:1986, Blondel:2005, Nedic:2009, Chen:2012, Patterson:2014} will be made to prove the correctness of the asymptotic limit of CDF's in Thm. \ref{thm:thmB}. We remark that this assumption is not needed to prove asymptotic consensus (Thm. \ref{thm:thmA}).
\begin{assumption} (Strong Connectivity over Interval) \label{assump:bounded_times}
	There exists $R$ such the graph $(\mathcal{N},E(t) \cup E(t+1) \cup\dots\cup E(t+R-1))$ is the same strongly connected graph for all $t$.
\end{assumption}
This condition is true with high probability as $R$ gets large.

%

\subsection{Analysis}
The principal component that enables the proofs of convergence of Algorithm 1 is Equation (\ref{eq:randomized_density_update}), that propagates the vector of belief functions forward in time. In (\ref{eq:rand_density_evol}), $\bA_{t}=\bA_{i_t\to j_t}$ is the time-varying interaction matrix between agents $i_t$ and $j_t$ and $\bD_{t}(x)=\bD_{i_t\to j_t}(x)$ is a diagonal time-varying matrix dependent on the response of agent $i$, $y_{i,t+1}$, the query region $A_{i_t,t} \subset \mathcal{X}$ and the state $x \in \mathcal{X}$. For intuition, (\ref{eq:rand_density_evol}) can be written as a sum of two terms:
\begin{align*}
	&\bp_{t+1}(x) = \bp_t(x) - \be_i p_{i,t}(x) - \be_j p_{j,t}(x) \\
		&+ (\be_i+\be_j) \left( \alpha_i p_{i,t}(x) \frac{l_i(y_{i,t+1}|x,A_{i,t})}{\mathcal{Z}_{i,t}(y_{i,t+1})} + (1-\alpha_i) p_{j,t}(x) \right)
\end{align*}
The first term simply zeroes out the $i$ and $j$ components and the second term fills them in with the average between the updated belief of agent $i$ and the current belief of agent $j$. The rest of the components are left intact.

Proposition \ref{prop:Amat} provides bounds on the dynamic range of $\bA\bx$, where $\bx$ is any arbitrary vector (see Theorem 3.1 in \cite{Seneta:1981}). The coefficient of ergodicity of an interaction matrix $\bA$ is defined as \cite{Seneta:1981, Ipsen:2011}:
\begin{equation} \label{eq:coeff_ergodicity}
	\tau_1(\bA) \stackrel{\text{def}}{=} \frac{1}{2} \max_{i\neq j} \nn \bA^T(\be_i-\be_j)\nn_1 = \frac{1}{2} \max_{i\neq j} \sum_{l=1}^M |a_{i,l}-a_{j,l}|
\end{equation}
This coefficient satisfies $\tau_1(\cdot)\in [0,1]$. The most non-ergodic interaction matrix is the identity matrix $\bI_M$, for which $\tau_1(\bI_M)=1$ and there is no information sharing, and the other extreme is the full-rank matrix $\tau_1(\frac{1}{M}\b1\b1^T)=0$.

\begin{proposition} (Contraction Property of $\bA$) \label{prop:Amat}
	Assume $\bA=\{a_{i,j}\}$ is a $M\times M$ stochastic matrix. Let $\bx$ be an arbitrary non-negative vector. Then, we have for all pairs $(i,j)$:
	\begin{equation*}
		[\bA\bx]_i-[\bA\bx]_j \leq \tau_1(\bA) \left( \max_i x_i - \min_i x_i \right)
	\end{equation*}
\end{proposition}
Although $\bA_t$ is not irreducible, $\bA=\EE[\bA_t]$ is, which is used in Lemmas 2,3. Next, we recall a tight smooth approximation to the non-smooth maximum and minima operators. Similar results have appeared in Prop. 1 in \cite{Chen:2013} and p. 72 in \cite{Boyd:ConvexOptimization}.
\begin{proposition} \label{prop:lse} (Tight Smooth Approximation to Maximum/Minimum Operator)
	Let $\ba \in \RR^M$ be an arbitrary vector. Then, we have for all $\gamma > 0$:
	\begin{equation} \label{eq:lse_approx_max}
		\max_i a_i \leq \frac{1}{\gamma} \log\left( \sum_{i=1}^M e^{\gamma a_i} \right) \leq \max_i a_i + \frac{\log M}{\gamma}
	\end{equation}
	and
	\begin{equation} \label{eq:lse_approx_min}
		\min_i a_i \geq -\frac{1}{\gamma} \log\left( \sum_{i=1}^M e^{-\gamma a_i} \right) \geq \min_i a_i - \frac{\log M}{\gamma}
	\end{equation}
\end{proposition}

\begin{lemma} \label{lemma:lemmaA}
	Consider Algorithm \ref{alg:alg1}. Let $B \in \mathcal{B}(\mathcal{X})$. Then, we have:
	\begin{equation*}
		\EE\left[\int_B \bD_t(x) \bp_t(x) dx \Bigg| \mathcal{F}_t \right] = 0.
	\end{equation*}
	where $\bD_t(x)$ was defined in (\ref{eq:rand_D}).
\end{lemma}
\begin{proof}
	See Appendix A.
\end{proof}

\begin{lemma} \label{lemma:lemmaB}
	Consider Algorithm \ref{alg:alg1}. Let $B \in \mathcal{B}(\mathcal{X})$. Then, we have $\EE[\bv^T\bP_{t+1}(B)|\mathcal{F}_t]=\bv^T\bP_t(B)$ for some positive vector $v\succ 0$, and $\lim_{t\to\infty} \bv^T\bP_{t}(B)$ exists almost surely.
\end{lemma}
\begin{proof}
	See Appendix B.
\end{proof}
Lemmas \ref{lemma:lemmaA} and \ref{lemma:lemmaB} imply that the term $\int_B \bv^T\bD_t(x)\bp_t(x) dx$ is a martingale difference noise term.

\begin{lemma} \label{lemma:lemmaC}
	Consider Algorithm \ref{alg:alg1} and let $d=1$. Let $B=[0,b] \in \mathcal{B}(\mathcal{X})$. Let $\bv$ denote the positive left eigenvector of $\bA=\EE_{i,j}[\bA_{i\to j}]$. Define the variable:
	\begin{align}
		\Lambda_t(B,P,\bepsilon) &\stackrel{\text{def}}{=} \frac{1}{e^{\bv^T\bP_t(B)}} \Bigg( \sum_{i,j=1}^M q_i P_{i,j} e^{\bv^T\bA_{i\to j}\bP_t(B)} \nonumber \\
			&\times \cosh((v_i+v_j)a_{i,i}(1-2\epsilon_i) \mu_{i,t}(B)) \Bigg) \label{eq:Lambda}
	\end{align}
	where $$\mu_{i,t}(B) \stackrel{\text{def}}{=} \min\{P_{i,t}(B), 1-P_{i,t}(B)\}.$$ Then, we have $\Lambda_t \stackrel{a.s.}{\longrightarrow} 1$ as $t\to\infty$.
\end{lemma}
\begin{proof}
	See Appendix C.
\end{proof}

\begin{lemma} \label{lemma:lemmaX}
	Consider the same setup as Lemma \ref{lemma:lemmaC}. Then, $\Lambda_t(B) \stackrel{a.s.}{\longrightarrow} 1$ implies $\mu_{i,t}(B) \stackrel{a.s.}{\longrightarrow} 0$ for all $i \in \mathcal{N}$ as $t\to\infty$.
\end{lemma}
\begin{proof}
	See Appendix D.
\end{proof}

Define the dynamic range (with respect to all agents in the network) of the posterior probability that $X^*$ lies in set $B\subset \mathcal{X}$:
\begin{equation} \label{eq:dyn_range}
	V_t(B) \stackrel{\text{def}}{=} \max_i P_{i,t}(B) - \min_i P_{i,t}(B)
\end{equation}
Also, define the innovation:
\begin{equation*}
	d_{i,t+1}(B) \stackrel{\text{def}}{=} \left[ \int_B \bD_t(x)\bp_t(x) dx \right]_i = \int_B [\bD_t(x)]_{i,i} p_{i,t}(x) dx
\end{equation*}

We next prove a lemma that shows that the dynamic range $V_t(B)$ has a useful upper bound. 
\begin{lemma} \label{lemma:lemmaD}
	Consider Algorithm \ref{alg:alg1}. Let $B=[0,b]$ with $b \leq 1$. Then, for all $R\in \NN$:
	\begin{align}
		V_{t+R}(B) &\leq \tau_1(\bA_{t+R-1}\cdots \bA_{t}) V_t(B) \nonumber \\
		&\quad + \sum_{k=0}^{R-1} \left( \max_i d_{i,t+R-k}(B) - \min_i d_{i,t+R-k}(B) \right) \label{eq:V_bound}
	\end{align}
\end{lemma}
\begin{proof}
	See Appendix E.
\end{proof}

To show convergence of the integrated beliefs of all agents in the network to a common limiting belief, it suffices to show $V_t(B) \stackrel{i.p.}{\to} 0$. Theorem \ref{thm:thmA} shows convergence of asymptotic beliefs to a common limiting belief. The structure of the limiting belief is given in Theorem \ref{thm:thmB}.
\begin{theorem} \label{thm:thmA} 
	Consider Algorithm \ref{alg:alg1} and let the Assumptions 1-2 hold. Let $B=[0,b]$, $b\leq 1$. Then, consensus of the agents' beliefs is asymptotically achieved across the network:
	\begin{equation*}
		V_t(B) = \max_{i} P_{i,t}(B) - \min_i P_{i,t}(B) \stackrel{i.p.}{\longrightarrow} 0
	\end{equation*}
	as $t\to\infty$.
\end{theorem}
\begin{proof}
	See Appendix F.
\end{proof}

Theorem \ref{thm:thmA} establishes that Algorithm 1 produces belief functions that become identical over all agents. This establishes asymptotic consensus among the beliefs, i.e., that as time goes on all agents come to agreement about the uncertainty in the target state. It remains to show the limiting belief is in fact concentrated at the true target state $X^*$ (Thm. \ref{thm:thmB}).
\begin{lemma} \label{lemma:lemmaE}
	Consider Algorithm \ref{alg:alg1}. Assume $p_{i,0}(X^*)>0, \forall i \in \mathcal{N}$. Then, the posteriors evaluated at the true target state $X^*$ have the following asymptotic behavior:
	\begin{equation*}
		\liminf_{m\to\infty} \frac{1}{m} \sum_{i=1}^M c_i \log(p_{i,m R}(X^*)) \geq K \quad (a.s.)
	\end{equation*}
	for some $c_i>0$.
\end{lemma}
\begin{proof}
	See Appendix G.
\end{proof}

Now, we are ready to prove the main consistency result of the asymptotic beliefs. The proof is based on the consensus result of Theorem \ref{thm:thmA}.
\begin{theorem} \label{thm:thmB} 
	Consider Algorithm \ref{alg:alg1} and let Assumptions 1-3 hold. Let $B=[0,b]$, $b\leq 1$. Then, we have for each $i\in \mathcal{N}$:
		\begin{equation*}
			F_{i,t}(b) = P_{i,t}(B) \stackrel{i.p.}{\longrightarrow} F_{\infty}(b) =\left\{ \begin{array}{ll} 0, & b<X^* \\ 1, & b>X^* \end{array} \right.
		\end{equation*}
		as $t\to\infty$. In addition, for all $i\in \mathcal{N}$:
		\begin{equation} \label{eq:consistency}
			\check{X}_{i,t} \stackrel{\text{def}}{=} \int_{x=0}^1 x p_{i,t}(x) dx \stackrel{i.p.}{\longrightarrow} X^*
		\end{equation}
\end{theorem}
\begin{proof}
	Theorem 1 implies that for each agent $i$,
	\begin{equation} \label{eq:cdf_asymp}
		F_{i,t}(b) \stackrel{i.p.}{\to} F_{\infty}(b)
	\end{equation}
	as $t\to\infty$, where $F_{\infty}(b)$ is a common limiting random variable. To finish the proof, we show $F_{\infty}(b)$ is equal to the constant $I(b>X^*)$. Lemma \ref{lemma:lemmaE} implies that $\sum_{i=1}^M c_i \log p_{i,mR}(X^*) \to \infty$ almost surely as $m\to\infty$. Thus, there exists an agent $i_0 \in \mathcal{I}_1$ such that $p_{i_0,mR}(X^*) \to\infty$. Lemma \ref{lemma:lemmaX} implies $\mu_{i_0,t}(b')=\min\{F_{i_0,t}(b'),1-F_{i_0,t}(b')\} \stackrel{a.s.}{\to} 0$ for any $b'\in [0,1]$. This asymptotic result, combined with the monotonicity of the CDF operator $F_{i_0,t}(\cdot)$ and $p_{i_0,mR}(X^*)\stackrel{a.s.}{\to} \infty$ imply $F_{i_0,mR}(b)\to I(b>X^*)$. It then follows from (\ref{eq:cdf_asymp}) that $F_{i,t}(b) \stackrel{i.p.}{\to} F_{\infty}(b)=I(b>X^*)$ for all $i\in \mathcal{N}$. The second part (i.e., (\ref{eq:consistency})) of the proof is identical to the one of Theorem 2 in \cite{OnDecEstActQ:2015}.
\end{proof}

\section{Experimental Validation} \label{sec:simulations}

This section presents simulations that validate the theory in Section \ref{sec:convergence} and demonstrate the benefits of the proposed asynchronous decentralized 20 questions algorithm. We compare the performance of the asynchronous algorithm proposed in this paper with the synchronous counterpart studied in \cite{OnDecEstActQ:2015}, the centralized fully Bayesian estimator implemented via the basic equivalence principle derived in \cite{Tsiligkaridis:TIT:2014} and the standard query-based estimator with no information sharing. The root mean-squared error (RMSE) was chosen as a performance metric and we consider the performance for strongly connected geometric random graphs \cite{GuptaKumar:2000} to model ad-hoc wireless network topologies. An example geometric random graph is shown in Fig. \ref{fig:grg_sample}.
\begin{figure}[ht]
	\centering
		\includegraphics[width=0.40\textwidth]{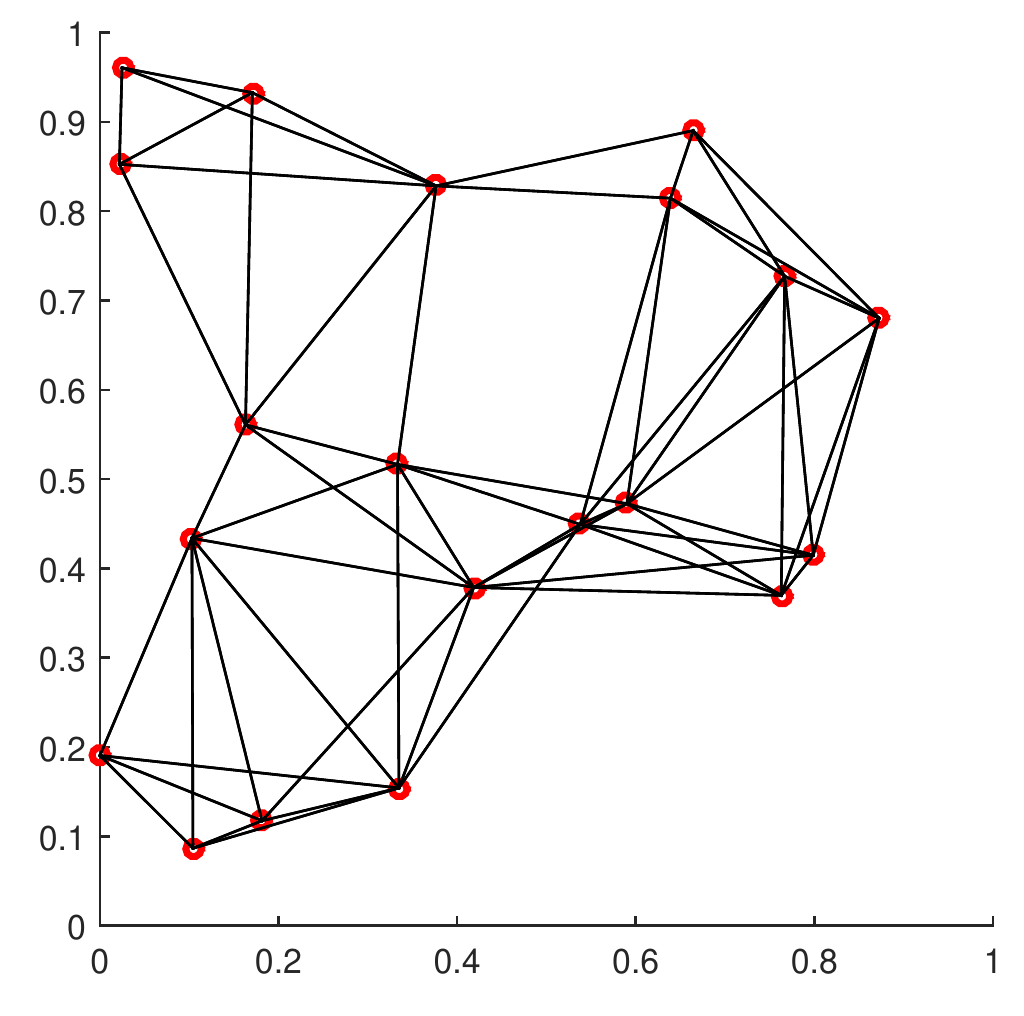}
		\caption{ Geometric random graph with $M=20$ nodes over unit square. This topology defines the zero pattern of the stochastic interaction matrix $\bP$. The nonzero probabilities $P_{i,j}$ for each row/node $i$ are chosen to be uniform, i.e., each agent is equally likely to communicated with any of its neighbors. }	
	\label{fig:grg_sample}
\end{figure}

To ensure fairness in the comparisons, each algorithm iteration consists of $M$ queries. Thus, the performance at each iteration of the synchronous setup (in which all agents are queried about the target location) should be compared with the performance at every $M$ iterations of the asynchronous setup. This is denoted as ``effective iteration'' in the figures. We average performance over $T = 200$ Monte Carlo trials for each random graph realization. Performance is further averaged over $10$ geometric random graphs to obtain the ensemble-average RMSE's. The average and worst-case RMSE metrics were calculated as:
\begin{align*}
	\text{RMSE}_{\text{avg}} &= \sqrt{\frac{1}{T} \sum_{t=1}^T \frac{1}{M} \sum_{i=1}^M (\hat{X}_{i,t}-X^*)^2 } \\
	\text{RMSE}_{\text{max}} &= \sqrt{\frac{1}{T} \sum_{t=1}^T \max_i (\hat{X}_{i,t}-X^*)^2 }
\end{align*}

Figures \ref{fig:rmse:avg:homogeneous} and \ref{fig:rmse:max:homogeneous} show the average and worst-case RMSE performance of the network as a function of iteration for the case of $M=20$ homogeneous agents, i.e., all agents are unreliable with error probability $\epsilon=0.45$. We observe that the asynchronous and synchronous algorithms both uniformly outperform the case of no information sharing over all iterations. We remark that the asynchronous algorithm seems to have a slower asymptotic rate of convergence as compared to the synchronous algorithm, while it seems to improve the RMSE more than the synchronous counterpart for the first few iterations. This may be due to biasing effects that occur in the synchronous algorithm; i.e., the belief is perturbed by multiple unreliable neighbors at each update step. In the asynchronous algorithm on the other hand, there are less perturbations in the initial learning stage since it consists of pairwise belief averaging at each update step.
\begin{figure}[ht]
	\centering
		\includegraphics[width=0.45\textwidth]{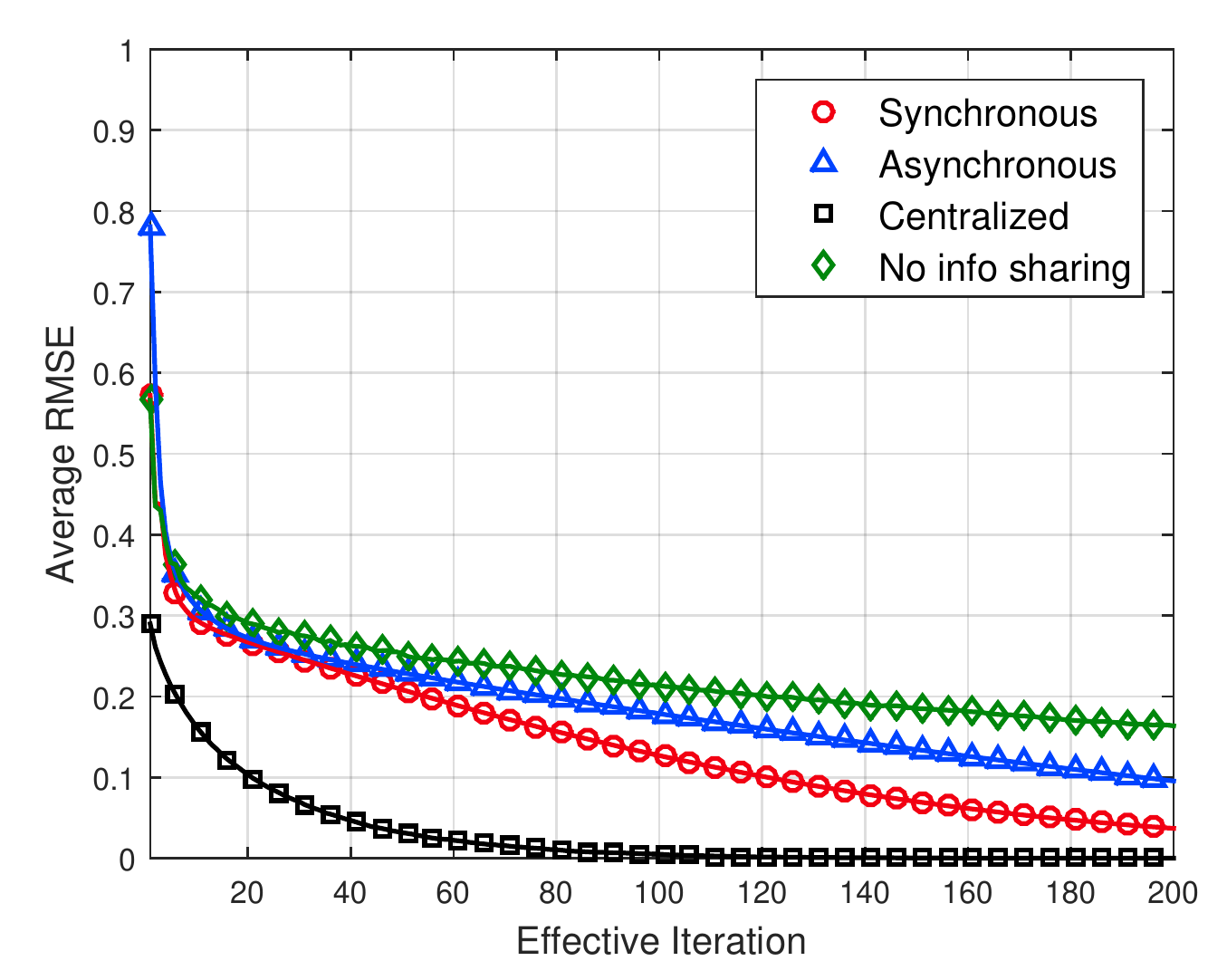}
		\caption{ Average RMSE performance across the network. Homogeneous network of $M=20$ unreliable agents with error probability $\epsilon=0.45$. The average RMSE is lower for the case of information sharing vs. the case of no information sharing.  }	
	\label{fig:rmse:avg:homogeneous}
\end{figure}
\begin{figure}[ht]
	\centering
		\includegraphics[width=0.45\textwidth]{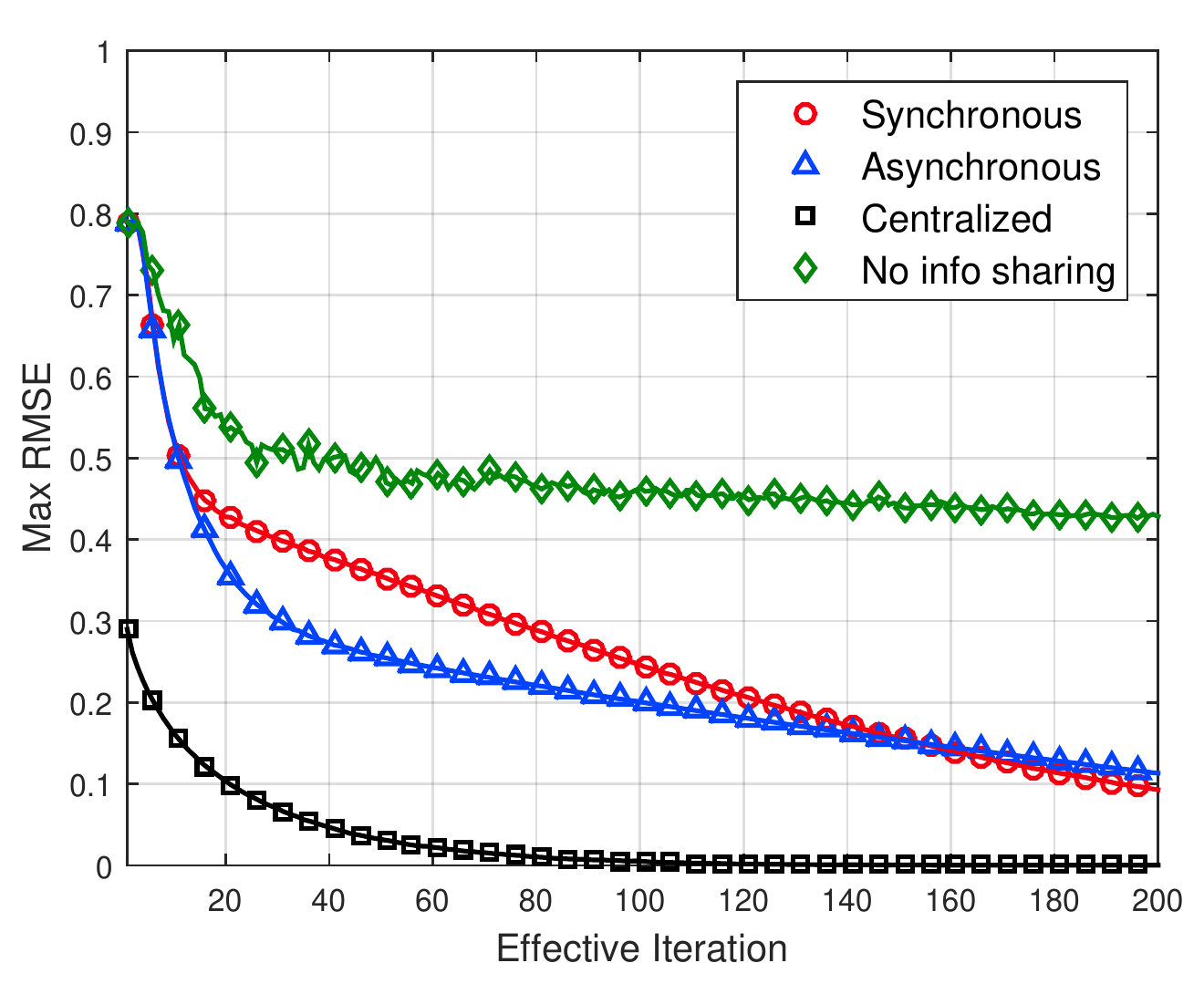}
		\caption{ Worst-case RMSE performance across the network. Homogeneous network of $M=20$ unreliable agents with error probability $\epsilon=0.45$. The worst-case MSE across the network is lower for the case of information sharing vs. the case of no information sharing. The asynchronous algorithm outperforms the synchronous algorithm in the initial learning stage. }	
	\label{fig:rmse:max:homogeneous}
\end{figure}

Figures \ref{fig:rmse:avg:heterogeneous} and \ref{fig:rmse:max:heterogeneous} show the average and worst-case RMSE performance of the network for the case of $M=20$ heterogeneous agents, i.e., three agents are reliable with error probability $\epsilon=0.05$ and the rest are unreliable with error probability $\epsilon=0.45$. Here, the reliable agents speed up the convergence of the unreliable agents through belief averaging. We observe the interesting result that the asynchronous algorithm uniformly outperforms the synchronous algorithm over all iterations, both in terms of average and worst-case RMSE. This can be attributed to the fact that one-way updating and pairwise averaging induce less bias when combining beliefs at each node at each update step, effectively spreading the good information around the network in a simplified fashion without being influenced too much by multiple neighbors. Thus, asynchronous implementations of the non-Bayesian decentralized 20 questions algorithm have the potential to improve network-wide estimation performance and getting closer to the centralized Bayesian performance, in addition to requiring significantly less infrastructure and computational complexity, in comparison to synchronous implementations.
\begin{figure}[ht]
	\centering
		\includegraphics[width=0.45\textwidth]{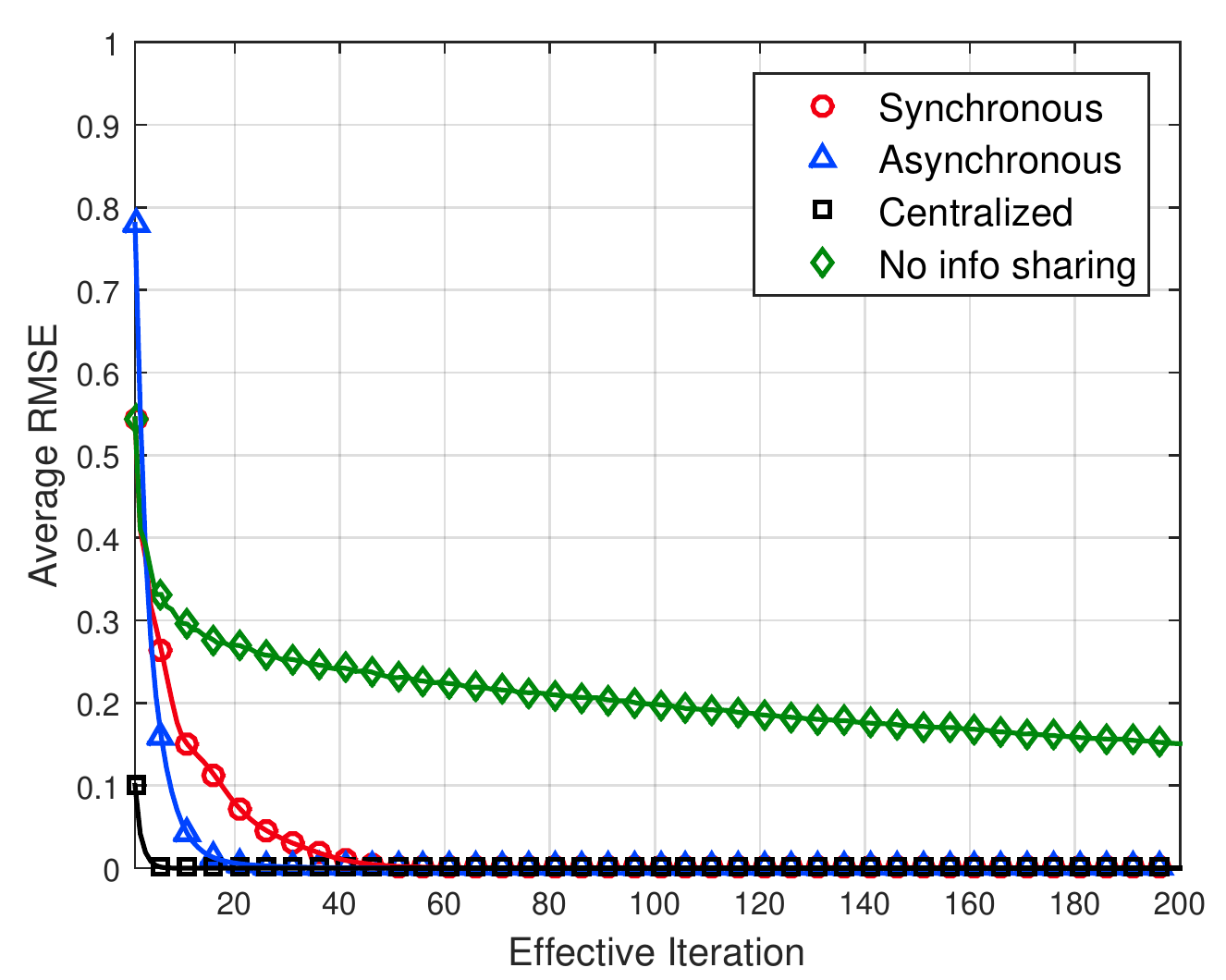}
		\caption{ Average RMSE performance across the network. Heterogeneous network of $M=20$ agents, three of which are reliable with error probability $\epsilon=0.05$ and the remaining ones are unreliable with error probability probability $\epsilon=0.45$. The asynchronous algorithm outperforms the synchronous decentralized estimation algorithm with information sharing, and the algorithm with no information sharing.  }
	\label{fig:rmse:avg:heterogeneous}
\end{figure}
\begin{figure}[ht]
	\centering
		\includegraphics[width=0.45\textwidth]{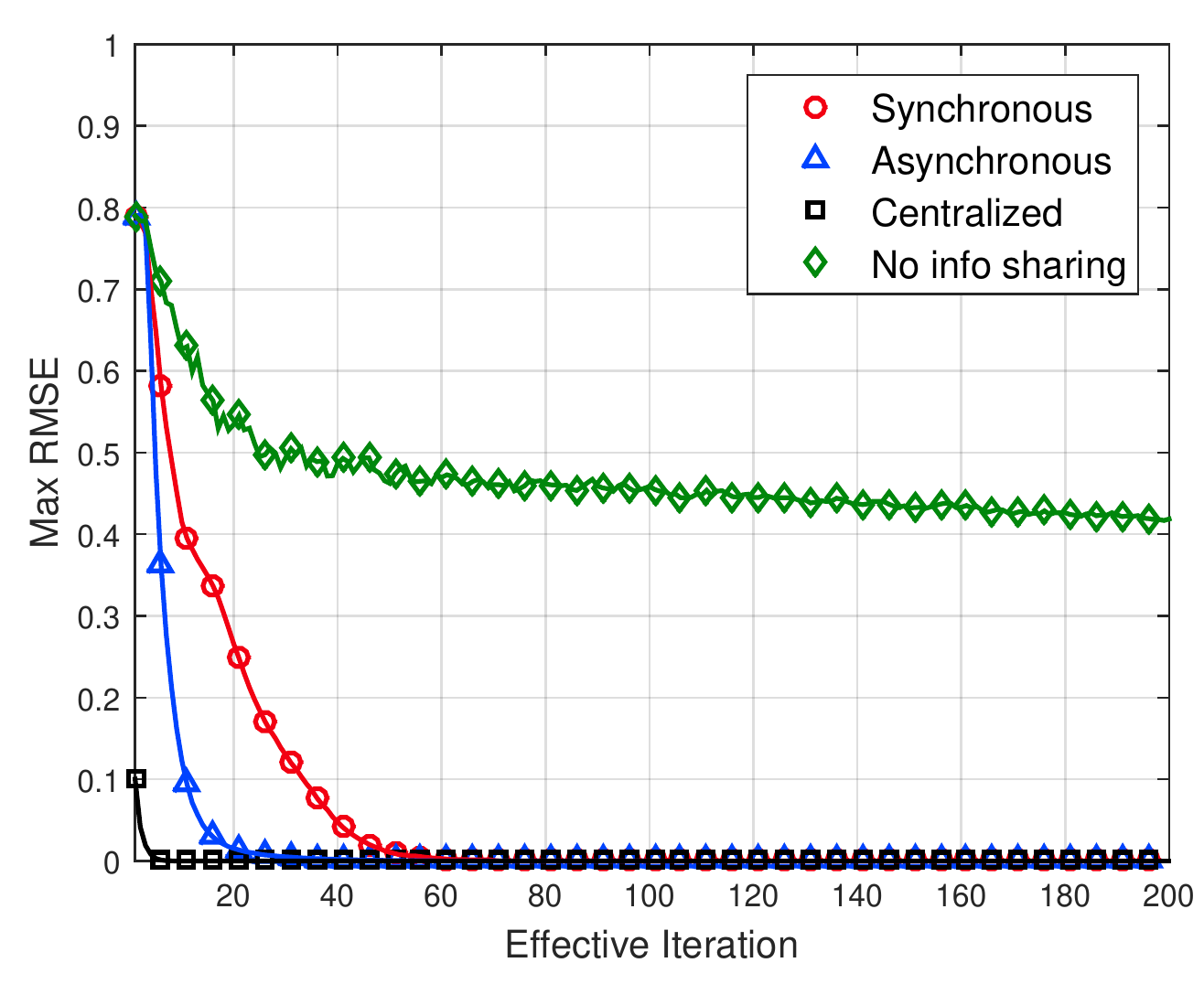}
		\caption{ Worst-case RMSE performance across the network. Heterogeneous network of $M=20$ agents, three of which are reliable with error probability $\epsilon=0.05$ and the remaining ones are unreliable with error probability probability $\epsilon=0.45$. The asynchronous algorithm outperforms the synchronous decentralized estimation algorithm with information sharing, and the algorithm with no information sharing. }	
	\label{fig:rmse:max:heterogeneous}
\end{figure}

\section{Conclusion} \label{sec:conclusions}
We introduced an asynchronous version of decentralized 20 questions with noise based on one-way updating and pairwise belief averaging, and analyzed its convergence properties. We also illustrated several benefits of information sharing as compared to no information sharing, and asynchronous vs. synchronous implementations. Asymptotic convergence properties of the agents' beliefs were derived, showing that they reach consensus to the true belief. Numerical experiments were presented to validate the convergence properties of the algorithm.

\newpage

\appendices

\section{Proof of Lemma \ref{lemma:lemmaA}}
\begin{proof}
	Note the following:
	\begin{align*}
		\EE &\left[ \int_B \bD_t(x) \bp_t(x) dx \Big| \mathcal{F}_t \right] \\
			&= \EE_{i,j}\left[ \EE\left[ \int_B \bD_{i\to j}(x) \bp_t(x) dx \Big| \mathcal{F}_t, (i,j) \right] \right]
	\end{align*}
	Decomposing the inner term, we obtain:
	\begin{align*}
		&\EE\left[ \int_B \bD_{i\to j}(x) \bp_t(x) dx \Big| \mathcal{F}_t, (i,j) \right] \\
			&= \EE\left[ \int_B \alpha_i \left( \frac{l_i(y_{i,t+1}|x,A_{i,t})}{\mathcal{Z}_{i,t}(y_{i,t+1})}-1 \right) (\be_i+\be_j) p_{i,t}(x) dx \Big| \mathcal{F}_t, (i,j) \right] \\
			&= \alpha_i (\be_i+\be_j) \EE\left[ \int_B (2l_i(y_{i,t+1}|x,A_{i,t}) - 1) p_{i,t}(x) dx \Big| \mathcal{F}_t,(i,j) \right]
	\end{align*}
	where we used the fact that $\mathcal{Z}_{i,t}(y)=1/2$ for all $y\in \mathcal{Y}$. This follows from the probabilistic bisection property:
	\begin{align*}
		&\mathcal{Z}_{i,t}(y) \\
			&= \int_{\mathcal{X}} p_{i,t}(x) \left( f_1^{(i)}(y) I(x\in A_{i,t}) + f_0^{(i)}(y) I(x\notin A_{i,t}) \right) dx \\
			&= f_1^{(i)}(y) P_{i,t}(A_{i,t}) + f_0^{(i)}(y) (1-P_{i,t}(A_{i,t})) = 1/2
	\end{align*}
	where we used the fact $f_1^{(i)}(y) + f_0^{(i)}(y) = 1$. From the definition of $l_i(y|x,A_{i,t})$, it follows that:
	\begin{align*}
		 &\int_B (2l_i(y_{i,t+1}|x,A_{i,t}) - 1) p_{i,t}(x) dx \\
				&= 2 \left( f_1^{(i)}(y_{i,t+1}) \PP_{i,t}(B\cap A_{i,t}) + f_0^{(i)}(y_{i,t+1}) \PP_{i,t}(B\cap A_{i,t}^c) \right) \\
				&\quad - \PP_{i,t}(B)
	\end{align*}
	Taking the conditional expectation, and using $\EE[ f_1^{(i)}(y_{i,t+1})]=1/2$, we obtain:
	\begin{equation*}
		 \EE\left[ \int_B (2l_i(y_{i,t+1}|x,A_{i,t}) - 1) p_{i,t}(x) dx \Big| \mathcal{F}_t,(i,j) \right] = 0
	\end{equation*}
	This concludes the lemma.
\end{proof}

\section{Proof of Lemma \ref{lemma:lemmaB}}
\begin{proof}
	From the tower property of conditional expectation, we obtain:
	\begin{align}
		&\EE[ \bP_{t+1}(B) |\mathcal{F}_t] \nonumber \\
				&= \EE\left[ \int_B (\bA_t+\bD_t(x)) \bp_t(x) dx \Big| \mathcal{F}_t \right] \nonumber \\
				&= \EE_{i,j}\left[  \EE\left[ \int_B (\bA_t+\bD_t(x)) \bp_t(x) dx \Big| \mathcal{F}_t, (i,j)\right] \Big| \mathcal{F}_t\right] \nonumber \\
				&= \sum_{i,j=1}^M q_i P_{i,j} \EE\left[ \bA_{i\to j} \bP_t(B) + \int_B \bD_{i\to j}(x) \bp_t(x) dx \Big| \mathcal{F}_t,(i,j) \right] \nonumber \\
				&= \left(\sum_{i,j=1}^M q_i P_{i,j} \bA_{i\to j} \right) \bP_t(B) + \EE\left[ \int_B \bD_t(x) \bp_t(x) dx \Big| \mathcal{F}_t \right] \nonumber \\
				&= \bA \bP_t(B)    \label{eq:lemmaB_relation}
	\end{align}
	where we used the result of Lemma \ref{lemma:lemmaA}. From strong connectivity (i.e., Assumption \ref{assump:strong_connectivity}), it follows that $\bA$ is an irreducible stochastic matrix. Thus, there exists a left eigenvector $\bv \in \RR^M$ with strictly positive entries corresponding to a unit eigenvalue-i.e., $\bv^T=\bv^T\bA$ \cite{Berman:1979}. Multiplying both sides of (\ref{eq:lemmaB_relation}) from the left by $\bv^T$, we obtain $\EE[\bv^T\bP_{t+1}(B)|\mathcal{F}_t] = \bv^T\bP_t(B)$. Thus, the process $\{\bv^T\bP_t(B):t\geq 0\}$ is a martingale with respect to the filtration $\mathcal{F}_t$. We note that it is bounded below by zero and above by $\nn \bv \nn_1$ almost surely. From the martingale convergence theorem \cite{Billingsley:2012}, it follows that it converges almost surely.
\end{proof}

\section{Proof of Lemma \ref{lemma:lemmaC}}
\begin{proof}
	Define the tilted measure variable $\zeta_{t}(B)\stackrel{\text{def}}{=}\exp(\bv^T\bP_t(B))$. From Lemma \ref{lemma:lemmaB} and Jensen's inequality, it follows that
	\begin{equation*}
		\EE[\zeta_{t+1}(B)|\mathcal{F}_t] \geq \zeta_t(B)
	\end{equation*}
	so the process $\{\zeta_t(B):t\geq 0\}$ is a submartingale with respect to the filtration $\mathcal{F}_t$. From the proof of Lemma \ref{lemma:lemmaB}, it follows that $\zeta_t(B)$ is bounded a.s., so by the martingale convergence theorem \cite{Billingsley:2012}, it follows that $\lim_{t\to\infty} \zeta_t(B)$ exists and is finite almost surely. Define the ratio $\beta_{t+1}(B)\stackrel{\text{def}}{=}\frac{\zeta_{t+1}(B)}{\zeta_t(B)}$. As a result, we have from Lemma \ref{lemma:lemmaB}:
	\begin{align*}
		&\lim_{t\to\infty} \beta_{t+1}(B) \\
		&= \lim_{t\to\infty} \frac{e^{\bv^T\bA_t\bP_t(B)} \exp\left( \bv^T \int_B \bD_t(x)\bp_t(x) dx \right)}{e^{\bv^T\bP_t(B)}} \stackrel{a.s.}{=} 1
	\end{align*}
	Since the variables in the limit on the LHS are bounded a.s., i.e., $\left|\beta_{t+1}(B)\right|\leq e^{\nn\bv \nn_1}$, the dominated convergence theorem for conditional expectations \cite{Durrett:2005} implies:
	\begin{equation} \label{eq:eq2}
		\EE\left[ \beta_{t+1}(B) \Bigg| \mathcal{F}_t \right] \stackrel{a.s.}{\longrightarrow} 1
	\end{equation}
	as $t\to\infty$. The conditional expectation can be expanded as:
	\begin{align}
		&\EE[\beta_{t+1}(B)|\mathcal{F}_t] \nonumber \\
			&= \frac{1}{e^{\bv^T\bP_t(B)}} \EE\left[ e^{\bv^T\bA_t\bP_t(B)} e^{\bv^T\int_B \bD_t(x)\bp_t(x) dx} | \mathcal{F}_t \right] \nonumber \\
			&= \frac{1}{e^{\bv^T\bP_t(B)}} \sum_{i,j} q_i P_{i,j} \nonumber \\
			&\quad \times e^{\bv^T\bA_{i\to j}\bP_t(B)} \EE[e^{\bv^T\int_B \bD_t(x)\bp_t(x) dx} | \mathcal{F}_t, (i,j)] \label{eq:eq3}
	\end{align}
	Next, we analyze the ratio of exponentials for two separate cases. First, consider the case $P_{i,t}([0,b])=\int_{0}^b p_{i,t}(x) dx\leq 1/2$. Using the definition of $\hat{X}_{i,t}$, it follows that $b\leq \hat{X}_{i,t}$. This implies that $l_i(y|x,A_{i,t})=f_1^{(i)}(y)$ for all $x\leq b$. Using this fact and $P(Y_{i,t+1}=y|\mathcal{F}_t,(i,j))=1/2$:
	\begin{align}
		&\EE[e^{\bv^T\int_B \bD_t(x)\bp_t(x) dx} | \mathcal{F}_t, (i,j)] \nonumber \\
		&= \frac{1}{2} \left(e^{(v_i+v_j) \alpha_i (1-2\epsilon_i) P_{i,t}(B)} + e^{-(v_i+v_j) \alpha_i (1-2\epsilon_i) P_{i,t}(B)} \right) \nonumber \\
		&= \cosh((v_i+v_j) \alpha_i (1-2\epsilon_i) P_{i,t}(B) ) \label{eq:eq3_caseA}
	\end{align}
	where we used the fact that $(e^a+e^{-a})/2=\cosh(a)$. Second, considering the complementary case $P_{i,t}([0,b])>1/2$, it follows that $b>\hat{X}_{i,t}$, and:
	\begin{align}
		&\EE[e^{\bv^T\int_B \bD_t(x)\bp_t(x) dx} | \mathcal{F}_t, (i,j)] \nonumber \\
		&= \cosh((v_i+v_j) \alpha_i (1-2\epsilon_i) P_{i,t}(B^c) ) \label{eq:eq3_caseB}
	\end{align}
	Combining the two cases (\ref{eq:eq3_caseA}) and (\ref{eq:eq3_caseB}) and using the definition of $\mu_{i,t}(B)$, we obtain:
	\begin{align*}
		&\EE[e^{\bv^T\int_B \bD_t(x)\bp_t(x) dx} | \mathcal{F}_t, (i,j)] \\
		&= \cosh\left(v_i \alpha_i (1-2\epsilon_i) \mu_{i,t}(B) \right)
	\end{align*}
	The proof is completed by substituting this expression into (\ref{eq:eq3}).
\end{proof}

\section{Proof of Lemma \ref{lemma:lemmaX}}
\begin{proof}
Using the definition of $\Lambda_t$, Jensen's inequality and the eigenrelation $\bv^T=\bv^T\bA$:
\begin{align*}
	&\Lambda_t(B) \\
	  &= \frac{1}{e^{\bv^T\bP_t(B)}} \EE_{i,j}\left[ e^{\bv^T\bA_{i\to j}\bP_t(B)} \cosh((v_i+v_j)\alpha_i(1-2\epsilon_i)\mu_{i,t}(B)) \right] \\
		&= \frac{1}{e^{\bv^T\bP_t(B)}} \EE_{i,j}\Bigg[ \frac{1}{2} e^{\bv^T\bA_{i\to j}\bP_t(B)+(v_i+v_j)\alpha_i(1-2\epsilon_i)\mu_{i,t}(B)}  \\
		&\quad + \frac{1}{2} e^{\bv^T\bA_{i\to j}\bP_t(B)-(v_i+v_j)\alpha_i(1-2\epsilon_i)\mu_{i,t}(B)}  \Bigg] \\
		&\geq \frac{1}{e^{\bv^T \bP_t(B)}} \Bigg( \frac{1}{2} e^{\bv^T\bA\bP_t(B)+\EE_{i,j}[(v_i+v_j)\alpha_i(1-2\epsilon_i)\mu_{i,t}(B)]}  \\
		&\quad + \frac{1}{2} e^{\bv^T\bA\bP_t(B)-\EE_{i,j}[(v_i+v_j)\alpha_i(1-2\epsilon_i)\mu_{i,t}(B)]} \Bigg) \\
		&= \frac{1}{2} \frac{1}{e^{\bv^T\bP_t(B)}} \Bigg( e^{\bv^T\bP_t(B)} e^{\EE_{i,j}[(v_i+v_j)\alpha_i(1-2\epsilon_i)\mu_{i,t}(B)]} \\
		&\quad + e^{\bv^T\bP_t(B)} e^{-\EE_{i,j}[(v_i+v_j)\alpha_i(1-2\epsilon_i)\mu_{i,t}(B)]} \Bigg) \\
		&= \cosh(\EE_{i,j}[(v_i+v_j)\alpha_i(1-2\epsilon_i)\mu_{i,t}(B)]) \\
		&\geq \cosh(\EE_i[v_i \alpha_i (1-2\epsilon_i) \mu_{i,t}(B)]) \geq 1
\end{align*}
We also used the fact that $\frac{e^a+e^{-a}}{2}=\cosh(a) \geq 1$ for all reals $a$. Since $\Lambda_t(B) \stackrel{a.s.}{\longrightarrow} 1$, it follows from the bound above that $\cosh(\EE_i[v_i \alpha_i (1-2\epsilon_i) \mu_{i,t}(B)]) \stackrel{a.s.}{\longrightarrow} 1$. From the property of $\cosh(\cdot)$, it follows that
\begin{equation*}
	\sum_{i=1}^M q_i v_i \alpha_i (1-2\epsilon_i) \mu_{i,t}(B) \stackrel{a.s.}{\longrightarrow} 0
\end{equation*}
which further implies $(1-2\epsilon_i) \mu_{i,t}(B) \stackrel{a.s.}{\longrightarrow} 0$ for all $i$, since $\min v_i>0$. Furthermore, $\mu_{i,t}(B) \stackrel{a.s.}{\longrightarrow} 0$ for $i\in \mathcal{I}_1$. The proof is complete.
\end{proof}

\section{Proof of Lemma \ref{lemma:lemmaD}}
\begin{proof}
Integrating both sides of the recursion (\ref{eq:rand_density_evol}):
\begin{equation} \label{eq:update_B}
	\bP_{t+1}(B) = \bA_t \bP_t(B) + \bd_{t+1}(B)
\end{equation}
Unrolling (\ref{eq:update_B}) over $R$ steps:
\begin{align*}
	\bP_{t+R}(B) &= \bA_{t+R-1}\cdots \bA_{t} \bP_t(B) \\
		&\quad + \sum_{k=0}^{R-1} \bA_{t+R-1} \cdots \bA_{t+R-k} \bd_{t+R-k}(B)
\end{align*}
Since the product of stochastic matrices is also a stochastic matrix, Proposition \ref{prop:Amat} implies:
\begin{align*}
	&V_{t+R}(B) \\
		&= \max_i P_{i,t+R}(B) - \min_i P_{i,t+R}(B) \\
		&\leq \tau_1(\bA_{t+R-1}\cdots \bA_{t}) V_t(B) \\
		&\quad + \max_{i,j} \sum_{k=0}^{R-1} \Bigg( [\bA_{t+R-1} \cdots \bA_{t+R-k} \bd_{t+R-k}(B)]_i \\
		&\quad - [\bA_{t+R-1} \cdots \bA_{t+R-k} \bd_{t+R-k}(B)]_j \Bigg) \\
		&\leq \tau_1(\bA_{t+R-1}\cdots \bA_{t}) V_t(B) \\
		&\quad + \sum_{k=0}^{R-1} \Bigg( \max_i [\bA_{t+R-1} \cdots \bA_{t+R-k} \bd_{t+R-k}(B)]_i \\
		&\quad - \min_i [\bA_{t+R-1} \cdots \bA_{t+R-k} \bd_{t+R-k}(B)]_i \Bigg) \\
		&\leq \tau_1(\bA_{t+R-1}\cdots \bA_{t}) V_t(B) \\
		&\quad + \sum_{k=0}^{R-1} \left( \max_i d_{i,t+R-k}(B) - \min_i d_{i,t+R-k}(B) \right)
\end{align*}
\end{proof}

\section{Proof of Theorem \ref{thm:thmA}}
\begin{proof}
From Lemma \ref{lemma:lemmaD}, we obtain:
\begin{align}
	&\EE[V_{t+R}(B)|\mathcal{F}_t] \leq \EE[\tau_1(\bA_{t+R-1}\cdots \bA_{t})|\mathcal{F}_t] V_t(B) \nonumber \\
		&\quad + \sum_{k=0}^{R-1} \EE\left[ \max_i d_{i,t+R-k}(B) - \min_i d_{i,t+R-k}(B) \Bigg| \mathcal{F}_t \right]  \label{eq:expA}
\end{align}
We next show that the remainder is asymptotically negligible-i.e.,
\begin{equation} \label{eq:remainder_asymp_0}
	\sum_{k=0}^{R-1} \EE\left[ \max_i d_{i,t+R-k}(B) - \min_i d_{i,t+R-k}(B) \Bigg| \mathcal{F}_t \right] \to 0.
\end{equation}
Re-writing the remainder term:
\begin{align*}
	&\sum_{k=0}^{R-1} \EE\left[ \max_i d_{i,t+R-k}(B) - \min_i d_{i,t+R-k}(B) \Bigg| \mathcal{F}_t \right] \\
	&=\sum_{k=0}^{R-1} \EE\left[ \EE\left[\max_i d_{i,t+R-k}(B) - \min_i d_{i,t+R-k}(B)\Big|\mathcal{F}_{t+R-k-1} \right] \Bigg| \mathcal{F}_t \right]
\end{align*}
In order to show (\ref{eq:remainder_asymp_0}), it suffices to show:
\begin{align}
	\EE&\left[\max_i d_{i,t+R-k}(B) - \min_i d_{i,t+R-k}(B)\Big|\mathcal{F}_{t+R-k-1} \right] \nonumber \\
		&\leq 4 \max_i \{ \alpha_i (1-2\epsilon_i) \mu_{i,t+R-k-1}(B)\}  \label{eq:sufficient_asymp_0}
\end{align}
Note that if (\ref{eq:sufficient_asymp_0}) holds, then Lemma \ref{lemma:lemmaX} implies that $\mu_{i,t+R-k-1}(B) \stackrel{a.s.}{\to} 0$ as $t+R-k-1 \to \infty$, and as a result, (\ref{eq:remainder_asymp_0}) follows. Thus, we next focus on proving the bound (\ref{eq:sufficient_asymp_0}).

Using Proposition \ref{prop:lse}, we obtain for any $\gamma > 0$:
\begin{align}
	\EE &\left[ \max_i d_{i,t+R-k}(B) - \min_i d_{i,t+R-k}(B)\Big|\mathcal{F}_{t+R-k-1} \right] \nonumber \\
	&\leq \frac{1}{\gamma} \EE\Bigg[ \log\left(\sum_{l=1}^M \exp(\gamma d_{l,t+R-k}(B))\right) \nonumber \\
	&\qquad + \log\left(\sum_{l=1}^M \exp(-\gamma d_{l,t+R-k}(B))\right) \Bigg| \mathcal{F}_{t+R-k-1} \Bigg] \nonumber \\
	&\leq \frac{1}{\gamma} \Bigg[ \log\left(\sum_{l=1}^M \EE[\exp(\gamma d_{l,t+R-k}(B))|\mathcal{F}_{t+R-k-1}] \right) \nonumber \\
	&\qquad + \log\left(\sum_{l=1}^M \EE[\exp(-\gamma d_{l,t+R-k}(B))|\mathcal{F}_{t+R-k-1}] \right) \Bigg]   \label{eq:expB}
\end{align}
where we used Jensen's inequality and the linearity of expectation.

Next, note that the (conditional) moment generating functions of the innovation terms can be written as a weighted average of hyperbolic cosines:
\begin{align}
	&\EE[e^{\beta d_{l,t+R-k}(B)} | \mathcal{F}_{t+R-k-1}] \nonumber \\
		&= \sum_{i,j=1}^M q_i P_{i,j} \cosh\left( \beta \alpha_i (1-2\epsilon_i) [\be_i+\be_j]_l \mu_{i,t+R-k-1}(B) \right) \label{eq:cond_mgf_d}
\end{align}
for any $\beta\in \RR$. To prove (\ref{eq:cond_mgf_d}), set $t'=t+R-k-1$, and proceed similarly as in the proof of Lemma \ref{lemma:lemmaC}:
\begin{align*}
	\EE&[e^{\gamma d_{l,t'+1}(B)}|\mathcal{F}_{t'}] \\
		&= \EE_{i,j}[ \EE[ e^{\gamma d_{l,t'+1}(B)} | \mathcal{F}_{t'}, (i,j)] | \mathcal{F}_{t'}] \\
		&= \sum_{i,j} q_i P_{i,j} \EE\left[ e^{\gamma \int_B [\bD_{i\to j}(x) \bp_{t'}(x)]_l dx} \right] \\
		&= \sum_{i,j} q_i P_{i,j} \cosh\left( \gamma \alpha_i (1-2\epsilon_i) [\be_i+\be_j]_l \mu_{i,t'}(B) \right)
\end{align*}
Plugging (\ref{eq:cond_mgf_d}) into (\ref{eq:expB}) and using Proposition \ref{prop:lse} again, we obtain:
\begin{align*}
	&\EE\left[ \max_i d_{i,t'+1}(B) - \min_i d_{i,t'+1}(B) \Bigg| \mathcal{F}_{t'} \right] \\
		&\leq \frac{2}{\gamma} \log\left(\sum_{l=1}^M \sum_{i,j} q_i P_{i,j} \cosh\left( \gamma \alpha_i (1-2\epsilon_i) [\be_i+\be_j]_l \mu_{i,t'}(B) \right) \right) \\
		&\leq \frac{2}{\gamma} \log\left(\sum_{i=1}^M \sum_{i,j} q_i P_{i,j} \exp\left( \gamma \alpha_i (1-2\epsilon_i) [\be_i+\be_j]_l \mu_{i,t'}(B) \right) \right) \\
		&\leq \frac{2}{\gamma} \log\left(\sum_{i=1}^M \max_{i,j} \exp\left( \gamma \alpha_i (1-2\epsilon_i) [\be_i+\be_j]_l \mu_{i,t'}(B) \right) \right) \\
		&\leq \frac{2}{\gamma} \log\left(\sum_{i=1}^M \exp\left( 2\gamma \max_i\{\alpha_i (1-2\epsilon_i) \mu_{i,t'}(B)\} \right) \right) \\
		&\leq 4 \left( \max_i\left\{ \alpha_i (1-2\epsilon_i) \mu_{i,t'}(B) \right\} + \frac{\log M}{\gamma} \right)
\end{align*}
We have now proven (\ref{eq:sufficient_asymp_0}) by taking $\gamma \to \infty$ to tighten the bound.

Plugging (\ref{eq:sufficient_asymp_0}) into (\ref{eq:expA}):
\begin{align}
	\EE&[V_{t+R}(B)|\mathcal{F}_t] \leq \EE[\tau_1(\bA_{t+R-1}\cdots \bA_{t})|\mathcal{F}_t] V_t(B) \nonumber \\
		&+ 4 \sum_{k=0}^{R-1} \max_i \left\{ \alpha_i (1-2\epsilon_i) \mu_{i,t+R-k-1}(B) \right\} \label{eq:expC}
\end{align}
Note that from Assumption \ref{assump:bounded_times}, it follows that $\EE[\tau_1(\bA_{t+R-1}\cdots \bA_{t})|\mathcal{F}_t] \leq 1-\epsilon <1$ for some $\epsilon \in (0,1)$.

Lemma \ref{lemma:lemmaX} implies that $\mu_{i,t+R-k-1}(B) \stackrel{a.s.}{\to} 0$, for all $i\in \mathcal{I}_1$. Note that here we used the positivity of the weights $\alpha_i$ along with the fact that $\epsilon_i < 1/2$. Define the non-negative sequence
\begin{equation*}
	\delta_t^{(R)} \stackrel{\text{def}}{=} 4 \sum_{k=0}^{R-1} \max_i \left\{ \alpha_i (1-2\epsilon_i) \mu_{i,t+R-k-1}(B) \right\}
\end{equation*}
The above implies $\delta_t^{(R)} \stackrel{a.s.}{\to} 0$ as $t\to\infty$. Taking the unconditional expectation of both sides in (\ref{eq:expC}):
\begin{equation} \label{eq:expD}
	\EE[V_{t+R}(B)] \leq (1-\epsilon) \EE[V_t(B)] + \EE[\delta_t^{(R)}]
\end{equation}
where $\EE[\delta_t^{(R)}] \to 0$ by the dominated convergence theorem. Using induction on (\ref{eq:expD}), we obtain for all $t=kR$:
\begin{equation*}
	\EE[V_{t}(B)] \leq (1-\epsilon)^{t/R} \EE[V_0(B)] + \sum_{l=0}^{t/R-1} (1-\epsilon)^l \EE[\delta_{t-(l+1)R}^{(R)}]
\end{equation*}
Taking limits of both sides and using $\EE[V_0(B)]<\infty$:
\begin{align*}
	\limsup_{k\to\infty} \EE[V_{kR}(B)] &\leq \left(\lim_{k\to\infty} (1-\epsilon)^k \right) \EE[V_0(B)] \\
		&\quad + \lim_{k\to\infty} \sum_{l=0}^{k-1} (1-\epsilon)^l \EE[\delta_{(k-1-l)R}^{(R)}] = 0
\end{align*}
Thus, the subsequence $\EE[V_{kR}(B)]$ converges to zero since $V_{t}(B)\geq 0$. Since $\EE[V_{t+1}(B)] \leq \EE[V_t(B)] + \EE[\delta_t^{(1)}]$ for all $t\in \NN$ and $\EE[\delta_t^{(1)}]\to 0$, it follows that the whole sequence $\EE[V_t(B)]$ converges to zero. Markov's inequality further implies $V_t(B)\stackrel{i.p.}{\to} 0$. The proof is complete.

\end{proof}

\section{Proof of Lemma \ref{lemma:lemmaE}}
\begin{proof}
	Evaluating the density update at $x=X^*$:
	\begin{align*}
		&\bp_{t+1}(X^*) = \bA_{t}\bp_t(X^*) + \bD_t(X^*) \bp_t(X^*)  \\
			&= \sum_{l\neq i_t,j_t} p_{l,t}(X^*) \be_l \\
			&+ \left[(1-\alpha_{i_{t}}) p_{j_t,t}(X^*) + \alpha_{i_{t}}2 P(Y_{i_t,t+1}|Z_{i_t,t})p_{i_t,t}(X^*) \right] \\
			&\qquad \times (\be_{i_t}+\be_{j_t})
	\end{align*}
	where $Z_{i,t}=I(X^*\in A_{i,t})$ is the query input to the noisy channel and $P(Y_{i,t+1}|Z_{i,t})$ models the binary symmetric channel for the $i$th agent. Taking the logarithm of both sides and using Jensen's inequality, we obtain for a collaborating agent pair $(i_t,j_t)$:
	\begin{align*}
		&\log p_{i_t,t+1}(X^*) = \log [ (1-\alpha_{i_{t}}) p_{j_t,t}(X^*) \\
			&\qquad + \alpha_{i_{t}} (2 P(Y_{i_t,t+1}|Z_{i_t,t})) p_{i_t}(X^*) ] \\
			&\quad \geq (1-\alpha_{i_{t}}) \log p_{j_t,t}(X^*) + \alpha_{i_{t}} \log p_{i_t,t}(X^*) \\
			&\qquad + \alpha_{i_{t}} \log (2 P(Y_{i_t,t+1}|Z_{i_t,t}))
	\end{align*}
	Writing this in vector form with the understanding that the logarithm of a vector is taken component-wise:
	\begin{equation} \label{eq:lb_relation}
		\log \bp_{t+1}(X^*) \succeq \bA_t \log \bp_{t}(X^*) + \underline{\alpha} \log(2 P(Y_t|Z_t)) (\be_{i_t}+\be_{j_t})
	\end{equation}
	where we defined the response and query variables at time $t$ as $Y_{t}=Y_{i_t,t+1}$ and $Z_t=Z_{i_t,t}$.
	
	Define the matrix product $\bPhi_{t_1:t_2} \stackrel{\text{def}}{=} \bA_{t_2}\bA_{t_2-1}\cdots \bA_{t_1}$ for $t_1 \leq t_2$, and $\bPhi_{t_1:t_2}=\bI$ for $t_1>t_2$.
	Using induction on (\ref{eq:lb_relation}), we obtain:
	\begin{align}
		&\log \bp_{t+R}(X^*) \succeq \bPhi_{t:t+R-1} \log \bp_t(X^*) \nonumber  \\
		&\qquad + \underline{\alpha} \sum_{k=1}^{R} \bPhi_{t+R+1-k:t+R-1} (\be_{i_{t+R-k}} + \be_{j_{t+R-k}}) \nonumber \\
		&\qquad \qquad \times  \log(2 P(Y_{t+R-k}|Z_{t+R-k}))   \label{eq:lb_relation2}
	\end{align}
	Let $\bc \succ 0$ be the left-eigenvector of $\bPhi=\bPhi_{t:t+R-1}$ (cf. Assumption \ref{assump:bounded_times}). Define the positive constants $\underline{c} \stackrel{\text{def}}{=} \min_i c_i$, and $\alpha_R \stackrel{\text{def}}{=} \min_i \left( \min\{\alpha_i,1-\alpha_i\} \right)^R$. Left-multiplying (\ref{eq:lb_relation2}) by $\bc^T$:
	\begin{align*}
		&\bc^T \log \bp_{t+R}(X^*) \geq \bc^T \log \bp_t(X^*) \\ 
			&\qquad + \underline{\alpha} \cdot \underline{c} \cdot \alpha_R \sum_{k=1}^{R} \log(2 P(Y_{t+R-k}|Z_{t+R-k}))
	\end{align*}
	where we used the bound $\bc^T\bPhi_{t+R+1-k:t+R-1}\be_i \geq \underline{c} \cdot \alpha_R, \forall i$, for $k=1,\dots,R$. By induction, we have for $m\in \NN$:
	\begin{align*}
		&\bc^T \log \bp_{mR}(X^*) \geq \bc^T \log \bp_0(X^*) \\
			&\qquad + \underline{\alpha} \cdot \underline{c} \cdot \alpha_R \sum_{l=1}^{m} \sum_{k=1}^{R} \log(2 P(Y_{(m-l)R-k}|Z_{(m-l)R-k}))
	\end{align*}
	Then, using the strong law of large numbers (LLN):
	\begin{align*}
		&\liminf_{m\to\infty} \frac{1}{m} \bc^T \log \bp_{mR}(X^*) \\
		&\geq \underline{\alpha} \cdot \underline{c} \cdot \alpha_R \lim_{m\to\infty} \frac{1}{m} \sum_{l=1}^{m} \sum_{k=1}^{R} \log(2 P(Y_{(m-l)R-k}|Z_{(m-l)R-k})) \\
		&= \underline{\alpha} \cdot \underline{c} \cdot \alpha_R \cdot \EE\left[ \sum_{k=1}^{R} \log(2 P(Y_{R-k}|Z_{R-k})) \right] \\
		&\geq \underline{\alpha} \cdot \underline{c} \cdot \alpha_R \cdot C\left(\max_{i\in \mathcal{I}_1} \epsilon_i\right) =: K > 0
	\end{align*}
	where $C(\epsilon)$ is the channel capacity of a BSC with crossover probability $\epsilon$.
\end{proof}


%


\bibliographystyle{IEEEtran}
\bibliography{refs}

\end{document}